\documentclass[journal,twocolumn]{./IEEEtran}

\usepackage{amsmath}

\usepackage{color}
\usepackage[latin1]{inputenc}
\usepackage{subfigure}
\usepackage{epsfig}

\newtheorem{theorem}{Theorem}

\newcommand{\omt}[1]{}

\begin{document}

\title{Bounds on Stability and Latency\\ in Wireless Communication}
\author{%
Vicent Cholvi\thanks{%
Department of Computer Science, Universitat Jaume~I, Castell\'on, Spain.}
and
Dariusz R. Kowalski\thanks{%
Department of Computer Science, University of Liverpool, UK.}
}

\date{}

\maketitle

\begin{abstract}
In this paper, we study stability and latency of routing in wireless networks where it is assumed that no collision will occur.  Our approach is inspired by the adversarial queuing theory, which is amended in order to model wireless communication. More precisely, there is an adversary that specifies transmission rates of wireless links and injects data in such a way that an average number of data injected in a single round and routed through a single
wireless link is at most $r$, for a given $r\in (0,1)$. We also assume that the additional ``burst'' of data injected during any time interval and scheduled via a single link is bounded by a given parameter $b$.

Under this scenario, we show that the nodes following so called {\em work-conserving} scheduling policies, not necessarily 
the same, are guaranteed stability (i.e., bounded queues) and reasonably small data latency (i.e., bounded time on data delivery), 
for injection rates $r<1/d$, where $d$ is the maximum length of a routing path. Furthermore, we also show that such a bound is asymptotically optimal on $d$.
\end{abstract}



%

\section{Introduction}

In this paper, we consider a multihop wireless network where data is transmitted from its source node to its destination node through other intermediate nodes.

One crucial issue to characterize the performance of a networks is that of \emph{stability}. Roughly speaking, a communication network system is said to be stable if data waiting to be delivered (backlog) is finitely bounded at any 
single
time. The importance of such an issue is obvious, since if one cannot guarantee stability, then one cannot hope 
for ensuring
deterministic guarantees for most of the network performance metrics.
One such metric is {\em latency}, defined as the maximum time
for delivering data from its source to its destination, taken over all data
occurring in the routing process.

Whereas in the last few years much of the analysis of worst-case behavior of multihop wireline networks and scheduling policies has been performed using \emph{adversarial} models,
which
try to create as much trouble for the scheduling algorithm as possible~\cite{aafkll01,Ech:Uni}, 
only a few papers have been focussed on wireless networks. In~\cite{DBLP:journals/join/BorodinOR04}, Borodin et al. considered a model in which each node can transmit, at each time step, to all its neighbors, and show that the Nearest-to-Go scheduling policy is stable. They also showed that the Longest-in-System policy is unstable. Andrews et al.~\cite{DBLP:journals/talg/AndrewsZ07}, in a model in which a node can transmit to only one neighbor at 
a
time step, 
provided some fully distributed scheduling algorithms that ensure network stability, both when the routes are specified by the adversary and when they are chosen by the nodes.

Contrary to the previous papers, which assumed that data doesn't suffer collisions when several nodes transmit at the same time, in~\cite{1146398} Chlebus et al. studied stability of some distributed broadcast protocols. However, they assumed a scenario in which the transmission range of each node reaches all the other nodes. The maximum throughput, defined to mean the maximum rate for which stability is achievable, was studied by Chlebus et al.~\cite{DBLP:journals/dc/ChlebusKR09}. Anantharamu et al.~\cite{DBLP:conf/opodis/AnantharamuCR09} extended this work by studying the impact of limiting the adversary by assigning independent rates of injecting data 
to
each node.

In this paper, we study stability in a scenario formed by a multihop wireless network, where each node has a, possibly different, work-conserving scheduling policy. We say a scheduling policy is {\em work-conserving} if it cannot be idle as long as there is data queued to be transmitted. Many well-known scheduling policies like FIFO (First-In-First-Out), LIS (Longest-In-System), SIS (Shortest-In System), FTG (Farthest-To-Go), NTS (Nearest-To-Source), etc., are work-conserving policies, whereas other policies like Round-Robin, GPS (Generalized Processor Sharing), WFQ (Weighted Fair Queueing), etc., are non-work-conserving.

Our main result shows that a network with nodes following a work-conserving scheduling policy is stable provided the data injection rate is lower than $1/d$, being $d$ the largest number of links that data can cross in the network. Furthermore, we also show that such a bound is asymptotically optimal on $d$.

The rest of the paper is organized as follows. In Section~\ref{sec:model} we introduce our adversarial model and in
Section~\ref{sec:wc} we present the main results about stability and latency
of wireless communication in the specified model.

\section{The model}	
\label{sec:model}

We use a modified version of the wireless adversarial model proposed 
by Andrews et al.~\cite{DBLP:journals/talg/AndrewsZ07}. We consider a wireless multihop undirected network of $n$ nodes, where each node acts as both a transmitter and a receiver. When data is transmitted from its source node to its destination node and they are too far away from each other to communicate, data may go through other nodes as intermediate hops. Each node contains a queue for each outgoing link and uses it to store there data to be sent along the corresponding link. We assume that data is fluid-like (in the sense that the unit to transmit can be as small as needed), and that several pieces of data may be transmitted along one link in one time step. Furthermore, we assume that data units don't suffer collisions.\footnote{This can be achieved by making a specific channel assignment based on Time/Frequency/Code division or other methods for resolving contention in the data-link layer.} This feature is similar to the wireline adversarial model, that also doesn't take into account collisions between packets.

Time is divided into fixed slots. Each node can transmit at different  capacities in the interval $[0,1]$, which may or may not vary over time as a result of changing wireless channel conditions. We use $r_{ij}(t)$ to denote the rate at which node~$i$ can transmit to node~$j$ at time slot~$t$, also referred as \emph{transmission rate}. It is assumed that the transmission rate is defined over all pair of nodes, since $r_{ij}(t)$ can be set to zero if nodes~$i$ and~$j$ are too far away from each other to communicate directly. Furthermore, we assume that a node can transmit to only one neighbor at each time step. This is the main feature that distinguishes the wireless adversarial model from the wireline one, in which a node is allowed to transmit to several nodes at the same time step.

The time evolution is seen as a game between a \emph{scheduling queue policy} which decides, at each time step, which data must be transmitted (if any), and a bounded \emph{adversary} that governs both the \emph{data arrivals} and the \emph{channel conditions}, i.e., the transmission rates.

\emph{The adversary.} Regarding the data arrivals, at each time step the adversary injects a set of data into some of the nodes in the network. More 
precisely,
such an injection is defined by a pair of parameters $(b,r)$, where $b \geq 1$ is a natural number and $r$ 
satisfies
$0 \leq r < 1$. The parameter $b$ (usually called \emph{burstiness}) models the short bursts of data 
the adversary
can inject into the network. The parameter $r$ (called the \emph{injection rate}) models the long-term rate at which data can be injected into the network. The adversary is free to choose both the source and the destination node for any injected data. It also specifies the routing path from 
the source to the destination that data must follow. Paths don't include the same link more than once, and data is absorbed after traversing its route. 

The adversary also controls the quality of channels between nodes, trying to create as much trouble for the scheduling policy as possible, by means of specifying the transmission rates. At each time slot and for each node~$i$, the adversary sets up the values of the rate vector $(r_{i1}(t), r_{i2}(t), ..., r_{in}(t))$ before node~$i$ makes its scheduling decision. 
These rates are not know to the scheduling algorithm.

In order for stability to be feasible, it is necessary to impose some restrictions on the adversary so that it 
would not be able to
fully load any link a priori. More specifically, we require that the adversary satisfies the following \emph{admissibility condition}. Let $I_{ij}(t)$ represent the total amount of data that the adversary injects at time~$t$ and has link $i j$ on its path. We say that the adversarial injection is 
{\em admissible for rate~$r$ and burst~$b$} if there exist fractions $x_{ij}(t) \in [0,1]$ such that

\begin{equation}
\label{admissible1}
\sum_{j} x_{ij}(t) = 1, \;\; \forall i, \; \forall t
\end{equation}

\begin{equation}
\label{admissible2}
\sum_{t \in T_x} I_{ij}(t) \leq r \sum_{t \in T_x} r_{ij}(t) x_{ij}(t) + b, \;\; \forall ij, \; \forall T_x
\end{equation}
where $T_x$ denotes a consecutive sequence of $x$ time steps. One can 
view
$x_{ij}(t)$ as representing fractional decisions that indicate the assignment of data injected by the adversary that wishes to pass through node~$i$ at each time step. The admissibility condition of Eq.~\ref{admissible1} (combined with Eq.~\ref{admissible2}) says that the total size of such a data is, on average, at most~$r$.

\emph{Stability.} In order to formally define stability, we denote by $d_p$ the number of links that a data unit $p$ has to cross. Furthermore, we denote by $a_i^{p}$ and $f_i^{p}$ the time instants that~$p$ respectively arrives at and departs from  the $i$th node on its routing path, where $1 \le i \le d_p$. If~$p$ leaves its $i$th link in time step $f_i^{p}$, it will arrive at its $(i+1)$st queue at time step $a_{i+1}^{p} = f_i^{p}$. Finally, we denote by $Q_i^{p}$ the time~$p$  spends in the queue of the $i$th node on its path, i.e.,~$Q_i^{p}=f_i^{p}-a_i^{p}$. Let $Q=\max_{p,1 \le i \le d_p}Q_i^{p}$.

Given an adversary $\mathcal{A}$ (as defined above) and a scheduling protocol $\mathcal{P}$, we say a network $\mathcal{G}$ is \emph{stable} if $Q \leq \infty$~\cite{DBLP:journals/talg/AndrewsZ07}.


\section{Stability Conditions and Latency of Routing with
Work-Conserving Scheduling Policies}
\label{sec:wc}

In this section, we obtain 
a formula for
the threshold value on data injection rate guaranteeing stability in wireless networks with work-conserving scheduling policies (i.e., nodes cannot be idle as long as there data queued to be transmitted). Furthermore, we also estimate data latency for injection rates below this threshold value.


We remark that each node may have its own, possibly different, scheduling policy (FIFO, LIFO, Longest-in-System, etc.), as long as they are work-conserving. Furthermore, the scheduling policies don't need to know the quality of the transmission channels (i.e., the values of the rate vectors), since they only take care of deciding the order in which data is transmitted.

The following theorem provides a bound on the injection rate that guarantees network stability under any work-conserving scheduling
policies.

\begin{theorem}
\label{ref:workconservingtheorem}
Any network in which all queues use a, possibly different, work-conserving scheduling policy and data 
are
injected by a $(b,r)$-adversary, is stable 
for
$r < \frac{1}{d}$, where $d$ is the largest number of hops that any data unit 
traverses
in the network. 
Furthermore, 
data latency
is bounded from above by $d \frac{b \Delta}{1-r d}$, where $\Delta$ denotes the maximum number of neighbors a node can have.
\end{theorem}

\begin{proof}
The proof has two parts. First, we 
show
that if $r < \frac{1}{d}$ then the maximum time interval data takes to cross any link is bounded, which implies stability. Second, we prove that
data latency
is also upper bounded 
by $d \frac{b \Delta}{1-rd}$, provided the first condition on stability $r<\frac{1}{d}$ holds.

In what follows, we denote as $N(i)$ the set of nodes that are neighbors of node~$i$. We also note that $d=max_{p}\{d_p\}$.

\emph{Remark~1:} Note that we don't assume, a priori, whether the scenario formed by the network, the scheduling policy, and the adversary, is stable or not. Thus, if it is unstable, the time~$p$ takes to leave its $i$th queue could be infinite (i.e., $f_i^{p} = \infty$).

\emph{Remark~2:} Note that if $f_i^{p} = \infty$ (for some~$p$) then $Q=\infty$. However, we base our proof of finding under which conditions, $Q < \infty$ (which will automatically imply $f_i^{p} < \infty$).

Part~(1):  Let $p$ be a data unit that attains the maximum $Q$ (i.e., $Q_i^p =Q$) at the $i$th node on its path. We will call the queue in this node the {\em $i$th queue of data~$p$}.

Let $t_B$ be the oldest time step such that (1) $t_B < a_i^p$, and (2) in every step in $(t_B,a_i^p]$ the $i$th queue is non-empty. Hence, we have that during the interval $(t_B,f_i^p]$ the $i$th queue is non-empty.

Define $\phi_i^{p}$ as the set formed by all data units served by the $i$th queue during the interval $(t_B,f_i^p]$, and let $p^{*}$ be the oldest data unit in $\phi_i^{p}$ (i.e., $\forall p' \in \phi_i^{p}\; (a_1^{p'}\geq a_1^{p^{*}})$).  Hence, by the definition of $p^{*}$, all data in $\phi_i^{p}$ must have been injected during the interval $[a_1^{p^{*}}, f_i^{p}]$. 

Based on the above mentioned scenario and on the definition of the adversarial model, $Q_i^{p}=f_i^{p}-a_i^{p}$ 
is
bounded by the maximum number of data units injected 
during the interval $[a_1^{p^{*}},f_i^p - 1]$ 
(i.e., the worst-case scenario is: where all data injected 
since the
time instant $a_1^{p^{*}}$ until~$p$ is served, cross the $i$th node of $p$ and is scheduled 
before~$p$) 
minus the data served by the $i$th queue of $p$ during the interval $[t_B,a_i^p]$. 
Recall that in each step in the period $[t_B,a_i^p]$ the $i$th queue of $p$ is non-empty.
%
We have
\begin{eqnarray*}
&&f_i^{p}-a_i^{p} \\
&\leq&  
\hspace*{-1em} \sum_{j \in N(i)} \Big( r \sum_{t=a_1^{p^*}}^{f_i^p - 1}  r_{ij}(t) x_{ij}(t) + b\Big) - \sum_{j \in N(i)} \sum_{t=t_B}^{a_i^p} r_{ij}(t) x_{ij}(t)  \\
&=&  
\hspace*{-1em} \sum_{j \in N(i)} \Big(r \sum_{t=a_1^{p^*}}^{t_B - 1}  r_{ij}(t) x_{ij}(t) + r  \sum_{t=t_B}^{a_i^p}  r_{ij}(t) x_{ij}(t) \ + \\
& & 
\hspace*{-1em} r \sum_{t=a_i^p + 1}^{f_i^p -1}  r_{ij}(t) x_{ij}(t) + b - \sum_{t=t_B}^{a_i^p} r_{ij}(t) x_{ij}(t)\Big)
\end{eqnarray*}

Now, taking into account that  $r \leq 1$, we have
\begin{eqnarray*}
&&f_i^{p}-a_i^{p} \\
&\leq& 
\hspace*{-1em} \sum_{j \in N(i)} \Big(r  \sum_{t=a_1^{p^*}}^{t_B - 1}  r_{ij}(t) x_{ij}(t) +  r \sum_{t=a_i^p + 1}^{f_i^p -1}  r_{ij}(t) x_{ij}(t)+ b\Big)
\end{eqnarray*}

and taking also into account that $r_{ij} \leq 1$, we finally obtain
\begin{eqnarray*}
f_i^{p}-a_i^{p} &\leq& 
\hspace*{-1em} \sum_{j \in N(i)} \Big(r  \sum_{t=a_1^{p^*}}^{t_B - 1}  x_{ij}(t) + r \sum_{t=a_i^p + 1}^{f_i^p -1}  x_{ij}(t)+ b\Big) 
\end{eqnarray*}


Let $k$ be the hop number of $p^*$ when it arrives to the node where $p$ attains the maximum $Q$. Taking into account the first admissibility condition (Eq.~(\ref{admissible1})) we have that  $\sum_{j \in N(i)} x_{ij}(t) = 1$ for all~$t$, where $x_{ij}(t) \in [0,1]$. Therefore, 
\begin{eqnarray*}
&& f_i^p - a_i^p  \\
&\leq& 
\hspace*{-0.5em} r (t_B - a_1^{p^*})  + r (f_i^p - a_i^p -1)  + |N(i)|\cdot b  \\
&=& 
\hspace*{-0.5em} r (t_B - a_k^{p^*} + a_k^{p^*} - a_1^{p^*})  + r (f_i^p - a_i^p -1)  + |N(i)|\cdot b  \\
&=& 
\hspace*{-0.5em} r (t_B - a_k^{p^*}) + r (a_k^{p^*} - a_1^{p^*})  + r (f_i^p - a_i^p -1)  + \\
&& \hspace*{-0.5em} |N(i)|\cdot b
\end{eqnarray*}

Since $a_k^{p^*} \geq t_B$, then we have
\begin{eqnarray*}
f_i^p - a_i^p &\leq& 
\hspace*{-0.5em} r (a_k^{p^*} - a_1^{p^*})  + r (f_i^p - a_i^p -1)  + |N(i)|\cdot b  
\end{eqnarray*}

Since  $f_i^p - a_i^p=Q_i^p=Q$ and $a_k^{p^*}-a_1^{p^*} \leq (d-1)Q$,
and taken into account that $Q$ is the maximum time a data unit takes to cross a link, and $d-1$ is the maximum number of links a data unit crosses until reaching its last queue, we have that
\begin{equation}
\nonumber
\begin{split}
Q  &  \leq r Q (d-1) + r (Q -1)  + |N(i)|\cdot b  \\
Q & \leq r Q  d -r + |N(i)|\cdot b
\end{split}
\end{equation}
It follows that $Q < \infty$ for $r < 1/d$.

Part~(2): Consider a data unit $p$ that traverses a path with $d_p$ hops, where $d_p \leq d$. This network satisfies the 
property delivered in
Part~(1). 
Call $\delta_p$ 
the 
latency
of~$p$ and let $\Delta = \max_i |N(i)|$. 
From the above derivation we see that 
$\delta_p \leq d_p Q \leq d \frac{b \Delta -r}{1-r d} \leq  d \frac{b \Delta}{1-r d}$. \\
\end{proof}

\section{Tightness of the bounds.}

In~\cite{DBLP:journals/ton/BennettBCCB02}, Bennet et al. introduce a family of networks  intended to provide stability bounds in the \emph{wireline model} (we refer to the ``standard'' adversarial queueing model~\cite{aafkll01} for wireline networks). 

We denote as \emph{scenario} the combination of a concrete network and a concrete adversarial strategy.

\subsection{Description of the scenario in~\cite{DBLP:journals/ton/BennettBCCB02}}
\label{subsec:scenario}

The structure of the family of networks used in~\cite{DBLP:journals/ton/BennettBCCB02} is illustrated in Figure~1 (here, we present a slightly modified description of an equivalent network).

\paragraph{Network topology:}

The network is formed by a collection of identical building blocks, arranged in a tree structure of depth~$J$. At each building block there are $h-1$ nodes (where $h \geq 3$ represents the maximum number of hops a packet can traverse), each node having only \emph{one} queue.  Within each building block each node has $kh$ external inputs (i.e., coming from nodes located in some other building blocks in the lower level) and one internal input (i.e., coming from the preceding node of the same building blocks), except for the first one that has no internal input.

Furthermore, each node has one internal output that is connected to the subsequent node's internal input, except for the last one which forms the building block output which is connected to the external input of some node in a building block of higher level.

\paragraph{Traffic description:}

Every building block has one internal source of traffic called \emph{transit traffic} which, after traversing one internal node, feeds the next internal node, except for the last one which feeds the building block output. Also, each internal node is fed with $kh$ external sources of traffic called \emph{building block inputs} which, after traversing that internal node, is absorbed at the subsequent node or at the first node of some other building block in the upper level\footnote{At this point, we note that in~\cite{DBLP:journals/ton/BennettBCCB02} the authors say that \emph{traffic from the building block input dies in a data sink}. However, it is equivalent to say that they are absorbed after traversing one internal node (i.e., they are absorbed at the subsequent node or at the first node of some other building block in the upper level).}.

\begin{figure*}
\begin{center}
\subfigure[Structure of a single building block.]{\includegraphics[scale=0.7]{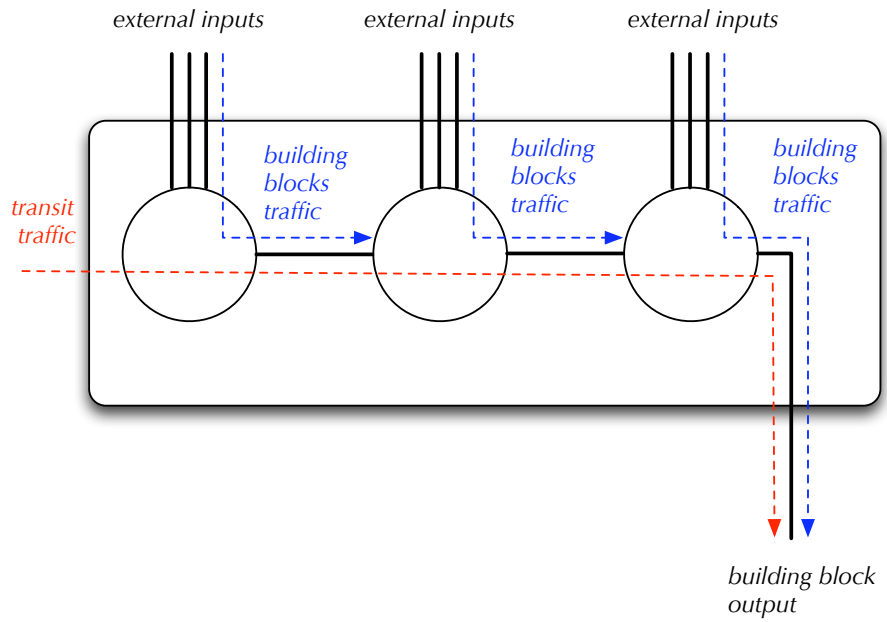}}
\label{aaa}
\subfigure[The network made of building blocks. For clarity, some links corresponding to the non-colored building blocks have been omitted.]{\includegraphics[scale=0.7]{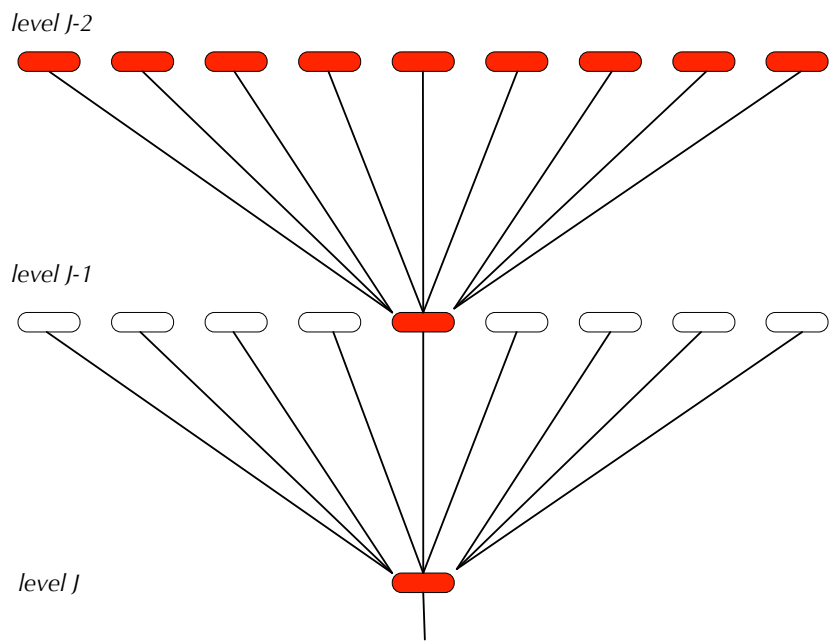}}
\label{bbb}
\caption{Network used in~\cite{DBLP:journals/ton/BennettBCCB02}.}
\label{fig:aaa}
\end{center} 
\end{figure*}

\subsection{Simulating the scenario in~\cite{DBLP:journals/ton/BennettBCCB02} in the wireless case}
\label{wireless:scenario}

In this section, we show how to simulate the scenario in~\cite{DBLP:journals/ton/BennettBCCB02} in the wireless model:

\begin{description} 
\item{S.1}
Each node uses a FIFO scheduling policy.

\item{S.2}
We use the same network topology as in~\cite{DBLP:journals/ton/BennettBCCB02}. That is, we use the same nodes and assume two nodes can exchange packets directly if they are connected in~\cite{DBLP:journals/ton/BennettBCCB02}.

\item{S.3}
The adversary in the wireless scenario sets up permanently the transmission rates to~$1$ for the nodes that  can exchange packets (as explained in S.2) and~$0$ for the remaining nodes.

\item{S.4}
The adversarial strategy of injecting packets and choosing their paths is exactly the same as in~\cite{DBLP:journals/ton/BennettBCCB02}.

\end{description}


In view of the network and traffic specification in Section~\ref{subsec:scenario} and by simulation assumptions S.2 and S.4, the following property holds:
\begin{description}
\item{Fact~1:}
Both in the scenario in~\cite{DBLP:journals/ton/BennettBCCB02} and the above described wireless scenario, each node feeds to  only  one node in the whole network.
\end{description}

\subsection{Summary of differences between the wireline scenario of~\cite{DBLP:journals/ton/BennettBCCB02} and the wireless model as described in Section~\ref{wireless:scenario}}

What distinguishes the collision-free wireless model from the wireline model of~\cite{DBLP:journals/ton/BennettBCCB02} is that, in the wireless model:
\begin{description}
\item{D.1}
A node feeds to only one neighbor at each time step.
\item{D.2}
The adversary can dynamically set up the transmission rates for the links.
\item{D.3}
The admissibility conditions are different (due to differences D.1 and D.2).
\end{description}

Let us analyze these differences in the contexts of the wireline scenario described in~\cite{DBLP:journals/ton/BennettBCCB02} and its simulation in the wireless model as described in Section~\ref{wireless:scenario}:
\begin{itemize}
\item
From Fact~1 we have that, in both scenarios, a node feeds to only one node in the whole network. So, there is no difference in both scenarios regarding D.1.
\item
Taking into account S.3, both scenarios are the same regarding D.2.
\item
Since the general differences D.1 and D.2 do not hold in the considered wireline and wireless scenarios, the admissibility conditions~(1) and~(2) in the wireless scenario are equivalent to the admissibility condition in the wireline scenario. Then,  there is no difference in both scenarios regarding D.3.
\end{itemize}

Therefore, we can conclude that the scenario in~\cite{DBLP:journals/ton/BennettBCCB02}  and the scenario introduced in Section~\ref{wireless:scenario} behave {\bf exactly} the same.

\subsection{Analysis of the stability bound}

In~\cite{DBLP:journals/ton/BennettBCCB02} it has been shown that if $r > 1/(d-1)$ a packet delay can be made larger than any fixed but arbitrary delay bound. This implies that any bound for stability must be lower or equal to  $1/(d-1)$, both in the wireline and the wireless models.  Let's refer to this bound as the \emph{optimistic} bound.


Now, we define a parameter $\epsilon(d)$ measuring the difference between such an \emph{optimistic} stability bound (i.e., $1/(d-1)$) and the bound provided in Theorem~1 (i.e., $1/d$) when we increase $d$. We have that it behaves like $\epsilon(d) \sim \frac{1}{d-1} - \frac{1}{d} \sim \frac{1}{d^2}$. Therefore, we have that our bound in Theorem~1 is asymptotically optimal within that limit.



\section{Future work}
As a future work, we note that an interesting open question is to analyze the behavior of the system when the channel conditions are not fully controlled by the adversary but they fulfill some specific constrains.
\bibliography{wireless_adv_full}
\bibliographystyle{IEEEtran.bst}

\end{document}